\documentclass[12pt, a4paper, reqno]{amsart}
\usepackage{mathrsfs}
\usepackage{bbm}
\usepackage{amsmath,amscd,amssymb,latexsym}
\usepackage[all]{xy}
\usepackage{color}
\usepackage[section]{placeins}
\usepackage[latin1]{inputenc}
\input xypic

\evensidemargin=\oddsidemargin
\DeclareMathVersion{can}
\DeclareMathAlphabet{\can}{OT1}{cmss}{m}{n}
\newtheorem{thm}{Theorem}[section]
\newtheorem{cor}[thm]{Corollary}

\newtheorem{lem}[thm]{Lemma}
\newtheorem{prop}[thm]{Proposition}
\newtheorem{rem}[thm]{Remark}

\theoremstyle{definition}
\newtheorem{defn}[thm]{Definition}
\theoremstyle{fact}

\theoremstyle{conjecture}

\numberwithin{equation}{section}

\begin{document}

\title[ ]{Extended Irreducible Binary Sextic Goppa codes}

\author[Huang] {Daitao Huang}
\address{\rm
College of Computer Science and Technology, Nanjing University of Aeronautics and Astronautics,
Nanjing,  211100, P. R. China}
\email{dthuang666@163.com}

\author[Yue]{Qin Yue}
\address{\rm Department of Mathematics, Nanjing University of Aeronautics and Astronautics,
Nanjing,  211100, P. R. China}
\email{yueqin@nuaa.edu.cn}

\thanks{This work was supported in part by National Natural Science Foundation of China (No. 61772015).  }




\begin{abstract}Let $n (>3)$ be a prime number and $\Bbb F_{2^n}$ a finite field of $2^n$ elements. Let $L =\Bbb F_{2^n}\cup \{\infty\}$ be the support set and $g(x)$ an irreducible polynomial  of degree $6$ over $\Bbb F_{2^n}$. In this paper,
we obtain an upper bound on the number of extended irreducible binary Goppa codes $\Gamma(L, g)$ of degree $6$ and length $2^n+1$.

\end{abstract}
\maketitle

\section{Introduction}
Let $\mathbb{F}_{2^n}$ be an extension of the finite field $\mathbb{F}_2$. Many codes defined over $\mathbb{F}_2$ are subfield subcodes of another code defined over the extension $\mathbb{F}_{2^n}$. For instance, Alternant codes are subfield subcodes of Generalized Reed-Solomon codes and
classical Goppa codes are a special case of Alternant codes \cite{16}.

Goppa codes are particularly appealing for cryptographic applications.
McEliece  was the first to exploit the potential of Goppa codes for the development of a secure cryptosystem \cite{9}.
Goppa codes are especially suited to this purpose because they have few invariants and the number of inequivalent codes grows exponentially with the length and dimension of the code.
Indeed, since the introduction of code-based cryptography by McEliece in 1978 \cite{18}, Goppa codes still remain among the few families of algebraic codes which resist to any structural attack. This is one of the reasons why every improvement of our knowledge of these codes is of particular interest.
In the McEliece cryptosystem, one chooses a random Goppa code as a key hence it is important that we know the number of Goppa codes for any given set of parameters. This will help in the assessment of how secure the McEliece cryptosystem is against an enumerative attack. An enumerative attack on the McEliece cryptosystem finds all Goppa codes for a given set of parameters and tests their equivalence with the public code \cite{15}.
There have been many attempts to count the number of Goppa codes.

In \cite{7,2,6}, it has been shown that  two Alternant codes $\mathcal{A}_k(v,L)$ and
 $\mathcal{A}_k(v',L')$ are equal even if the parameters $(v,L)$ are different from $(v',L')$. Based on these result, many extended Goppa codes become equivalent  by a parity check in \cite{11,1}. Moreno \cite{5} proved that there is only one extended irreducible binary degree $3$ Goppa code of any length and that there are four quartic Goppa codes of length $33$. Ryan \cite{1} focused on the quartic case of length $2^n+1$, where $n(>3)$ is prime, and gave a upper bound on the number of inequivalent irreducible Goppa codes. Some similar topic were investigated in \cite{14,12,13}.
In \cite{8}, Magamba and Ryan obtained an upper bound on the number of extended irreducible $q$-ary Goppa codes of degree $r$ and length $q^n+1$, where $p, r(>2)$ are two prime numbers and $q=p^t$.
In \cite{17}, Musukwa gave an upper bound on the number of inequivalent  extended irreducible binary Goppa codes of degree $2p$ and length $2^n+1$,  where $n$ and $p$ are two odd primes such that $p \neq n$ and $p\nmid (2^n\pm 1)$.

In this paper, we shall investigate  extended irreducible binary Goppa codes of degree $r$, where  $r$ is a product of two distinct prime numbers.  In the case: $r=6$,  we obtain upper bound of number of extended irreducible binary sextic Goppa codes of length $2^n+1$, where   $n(>3)$ is a prime.

The paper is organized as follows. In Section $2$, we remind some definitions of Alternant codes, Goppa codes, and extended Goppa codes. Moreover, we recall sufficient condition of equivalent extended Goppa codes by projective linear group acting on the projective line $\mathbb{F}_{2^n}\bigcup \{\infty\}$. In Section $3$, we investigate projective semi-linear group acting on the elelment of degree $6$ over $\mathbb{F}_{2^n}$ with $n(>3)$ a prime.
In Section $4$, we obtain the upper bound on the number of inequivalent extended Goppa codes by applying the Cauchy Frobenius Theorem. We conclude the paper in Section $5$.

\section{Preliminaries}
Let $\mathbb{F}_{2^n}$ be the finite field of order $2^n$ and $\overline{\mathbb{F}}_{2^n}=\Bbb F_{2^n}\cup\{\infty\}$ a set of coordinates for the projective line. We will introduce some basic knowledge in the following.
\subsection{Alternant code, Goppa Code and extended Goppa Code}
In the subsection, we describe concepts of Alternant codes, Coppa codes, and extended Goppa codes. In detail,  see \cite{6,1}.

\begin{defn}
Let $L=(\alpha_0,\ldots,\alpha_{m-1})$ be an $m$-tuple of distinct elements of $\mathbb{F}_{2^n}$,  $v=(v_0,\ldots,v_{m-1})$  an $m$-tuple of non-zero elements of $\mathbb{F}_{2^n}$, and  $r$  an integer less than $m$. The Alternant code $\mathcal{A}_{r}(v,L)$ is defined as follows:
 $$\mathcal{A}_{r}(v,L)=\{x=(x_0,\ldots,x_{m-1})\in \mathbb{F}_{2}^m : H_r(v,L)x^T=0\},$$
 where the  parity-check matrix is
$$H_r(v,L)=\left(
  \begin{array}{cccc}
    v_0 & v_1 & \ldots & v_{m-1} \\
    v_0\alpha_0 & v_1\alpha_1 & \ldots & v_{m-1}\alpha_{m-1} \\
    \vdots& \vdots &  & \vdots \\
    v_0\alpha_0^{r-1} & v_1\alpha_1^{r-1} & \ldots & v_{m-1}\alpha_{m-1}^{r-1} \\
  \end{array}
\right).$$
\end{defn}


\begin{defn}\label{suppp}
Let $L=(\alpha_0,\ldots,\alpha_{m-1})$ be an $m$-tuple of distinct elements of $\mathbb{F}_{2^n}$ and $g(x)\in \mathbb{F}_{2^n}[x]$  a polynomial of degree $r(< m)$ such that $g(\alpha_i)\neq 0$ for $i=0,1,\ldots,m-1$. The Goppa code $\Gamma(g,L)$  with the  Goppa polynomial $g(x)$ and the support $L$ is defined as follows:
$$\Gamma(g,L)=\left\{x=(x_0,x_2,\ldots, x_{m-1})\in \Bbb F_2^m: \sum_{i=0}^{m-1}\frac {x_i}{x-\alpha_i}\equiv 0\pmod {g(x)}\right\}.$$
If $g(x)$ is irreducible over $\mathbb{F}_{2^n}$,  $\Gamma(g,L)$ is called irreducible.
\end{defn}

\begin{rem} Let $\Gamma(g,L)$ be the Goppa code in Definition \ref{suppp}.

(1)  $\Gamma(g,L)=\mathcal{A}_r(v_{g,L},L)$, where  $v_{g,L}=(g(\alpha_0)^{-1},g(\alpha_1)^{-1},\ldots,g(\alpha_{m-1})^{-1})$.

(2) The expurgated Goppa code $\widetilde{\Gamma}(g,L)$ of  $\Gamma (g,L)$ is defined as follows:
$$\widetilde{\Gamma}(g,L)=\left\{x=(x_0,\ldots,x_{m-1})\in \Gamma(g,L): \sum\limits_{i=0}^{m-1}x_i=0\right\}.$$
In fact,  $\widetilde{\Gamma}(g,L)=\mathcal A_{r+1}(v_{g,L}, L)$, where  $v_{g,L}=(g(\alpha_0)^{-1},g(\alpha_1)^{-1},\ldots,g(\alpha_{m-1})^{-1})$.

(3) The extended Goppa code $\overline{\Gamma}(g,L)$ of $\Gamma(g,L)$ is defined as follows:
$$\overline{\Gamma}(g,L)=\left\{x=(x_0,\ldots,x_{m-1},x_m): (x_0,\ldots,x_{m-1})\in \Gamma(g,L), \sum\limits_{i=0}^{m}x_i=0\right\}.$$
\end{rem}

\begin{defn}
Let $L=(\alpha_0,\ldots,\alpha_{m-1})$ be an $m$-tuple of distinct elements of $\mathbb{F}_{2^n}$, $\overline{L}=L \bigcup \{\infty\}=(\alpha_0,\ldots,\alpha_{m-1},\infty)$, $v=(v_0,\ldots,v_m)$ an $(m+1)$-tuple of non-zero elements of $\mathbb{F}_{2^n}$, and $r$ an integer less than $m+1$. The Alternant code $\mathcal{A}_r(v,\overline{L})$ is defined as follows:
 $$\mathcal{A}_r(v,\overline{L})=\{x=(x_0,\ldots,x_{m})\in \mathbb{F}_2^{m+1}: H_r(v,\overline{L})x^T=0\},$$
where the parity-check matrix is
$$H_r(v,\overline{L})=\left(
  \begin{array}{cccc}
    v_0 & \ldots & v_{m-1} & 0 \\
    \vdots & & \vdots & \vdots \\
    v_0\alpha_0^{r-2} & \ldots & v_{m-1}\alpha_{m-1}^{r-2}& 0 \\
    v_0\alpha_0^{r-1} & \ldots & v_{m-1}\alpha_{m-1}^{r-1}& v_m \\
  \end{array}
\right).$$
\end{defn}


\begin{rem}\label{extn} Let $\Gamma(g,L)$ be the Goppa code in Definition \ref{suppp},
 $g(x)=\sum\limits_{i=0}^{r}g_ix^i$  a polynomial of degree $r$, and  $\overline L=L\cup \{\infty\}$.
Then the extended Goppa code $\overline{\Gamma}(g,L)$ of $\Gamma(g, L)$ is just the Altrnant code $\mathcal{A}_{r+1}(v_{g,\overline{L}},\overline{L})$, where $v_{g,\overline{L}}=(g(\alpha_0)^{-1},\ldots,g(\alpha_{m-1})^{-1},g(\infty)^{-1})$, $g(\infty)=g_r$.
\end{rem}

\subsection{Action of  groups}\label{agl}
We will recall the actions of the projective linear group  and the projective semi-linear group on $\mathbb{F}_{2^n}$ and $\overline{\mathbb{F}}_{2^n}=\Bbb F_{2^n}\cup \{\infty\}$, respectively.

There are some matrix  groups as follows:

(1) the affine group $$AGL_2(\mathbb{F}_{2^n})=\left\{M=\left(\begin{array}{cc}a &b \\ 0 & 1\end{array}\right):   a \in \mathbb{F}_{2^n}^*, b \in \mathbb{F}_{2^n}\right\};$$

(2) the general linear group
$$GL_2(\mathbb{F}_{2^n})=\left\{M=\left(
                               \begin{array}{cc}
                                 a & b \\
                                 c & d \\
                               \end{array}
                             \right)
 : a,b,c,d \in \mathbb{F}_{2^n}, ad-bc \neq 0\right\};$$

(3) the  projective linear group
$$PGL_2(\mathbb{F}_{2^n})=GL_2(\Bbb F_{2^n})/\{aE_2: a\in \Bbb F_{2^n}^*\},$$
where $E_2$ is the $2\times 2$ identity  matrix;

(4) the projective semi-linear group
$$P\Gamma L_2(\Bbb F_{2^n})=PGL_2(\Bbb F_{2^n})\times G,$$
where $G=Gal(\Bbb F_{2^{6n}}/\Bbb F_2)=\langle \sigma\rangle$ is the Galois group,  $\sigma(x)=x^2$ for $x\in \Bbb F_{2^{6n}}$, and its operation $\cdot$ is as follows:
$$(\widetilde A, \sigma^i)\cdot (\widetilde B, \sigma^j)=(\widetilde {A\cdot \sigma^i(B)}, \sigma^{i+j}), 0\le i, j\le 6n-1.$$

Let $M=\left(\begin{array}{cc} a&b\\c&d\end{array}\right)\in GL_2(\Bbb F_{2^n})$. Then the projective linear group $PGL_2(\Bbb F_{2^n})$ acts on $\overline {\Bbb F}_{2^n}$ as follows:
\begin{eqnarray*}PGL_2(\Bbb F_{2^n})\times \overline{\Bbb F}_{2^n}&\rightarrow& \overline {\Bbb F}_{2^n}\\
(\widetilde{M}, \zeta)&\mapsto& \widetilde M(\zeta)=M(\zeta)=\frac{a\zeta+b}{c\zeta+d},\end{eqnarray*}
where
$\frac{1}{0}=\infty$ and  $\frac{1}{\infty}=0$; the projective semi-linear group $P\Gamma L_2(\Bbb F_{2^n})$ acts on $\overline {\Bbb F}_{2^n}$ as follows:
\begin{eqnarray*}P\Gamma L_2(\Bbb F_{2^n})\times \overline{\Bbb F}_{2^n}&\rightarrow& \overline {\Bbb F}_{2^n}\\
((\widetilde{M},\sigma^i), \zeta)&\mapsto& (\widetilde M, \sigma^i)(\zeta)=M(\sigma^i(\zeta))=\frac{a\zeta^{2^i}+b}{c\zeta^{2^i}+d}.\end{eqnarray*}


In \cite{6}, there is a result as follows.

\begin{lem}\label{eq}
Let $g(x)$ be a polynomial of degree $r$ over $\mathbb{F}_{2^n}$ and $L=(\alpha_0,\ldots,\alpha_{m-1})$ an ordered tuples of $m$ distinct points in the projective line set $\overline{\mathbb{F}}_{2^n}$. Let $M=\left(\begin{array}{cc} a&b\\c&d\end{array}\right)\in GL_2(\Bbb F_{2^n})$, $L''=({M}(\alpha_0),\ldots,{M}(\alpha_{m-1}))$, $g'(x)=(-cx+a)^rg({M}^{-1} (x))$, $M^{-1}(x)= \frac{dx-b}{-cx+a}$,
and $g(-\frac dc)\ne 0$ if $c\ne0$. Then the Alternant code $\mathcal{A}_{r+1}(v_{g,L},L)$ is equal to the Alternant code  $\mathcal{A}_{r+1}(v_{g',L''},L'')$.
\end{lem}

By Lemma \ref{eq}, we have the following corollary.

\begin{cor}\label{2.777}Let $g(x)=\sum_{k=0}^rg_kx^k$ be a polynomial of degree $r$ over $\mathbb{F}_{2^n}$ and $L=(\alpha_0,\ldots,\alpha_{m-1})$ an ordered tuples of $m$ distinct points in the projective line set $\overline{\mathbb{F}}_{2^n}$. Let $M=\left(\begin{array}{cc} a&b\\c&d\end{array}\right)\in GL_2(\Bbb F_{2^n})$, $\tau\in G=Gal(\Bbb F_{2^{6n}}/\Bbb F_2)$,  $L'=(\tau(\alpha_0),\ldots,\tau(\alpha_{m-1}))$,   $L''=({M}(\tau(\alpha_0)),\ldots,{M}(\tau(\alpha_{m-1})))$,  $\tau g(x)=\sum_{k=0}^r\tau(g_k)x^k$,  $(\tau g)'(x)=(-cx+a)^r\cdot \tau g(M^{-1} (x))$, and $(\tau g)(-\frac dc)\ne 0$ if $c\ne 0$. Then the Alternant code $\mathcal{A}_{r+1}(v_{g,L},L)$ is equal to the Alternant code  $\mathcal{A}_{r+1}(v_{(\tau g)',L''},L'')$.
\end{cor}
\begin{proof} Since the Alernant code $\mathcal{A}_{r+1}(v_{g,L},L)$ is a subfield subcode, the Alternant code is equal to the Alternant code
$\mathcal A_{r+1}(v_{\tau g, L'}, L')$. By Lemma \ref{eq}, the Alternant code $\mathcal A_{r+1}(v_{\tau g, L'}, L')$ is equal to the Alternant code
$\mathcal A_{r+1}(v_{(\tau g)', L''}, L'')$.
\end{proof}

In Corollary \ref{2.777}, if $\tau$ is equal to the identity transformation, then Lemma \ref{eq} is direct from Corollary \ref{2.777}.
If
$g(x)$ is irreducible over $\Bbb F_{2^n}$  and has a root $\alpha$ in an extension over $\Bbb F_{2^n}$, then $(\tau g)'(x)=(-cx+a)^r\cdot \tau g(M^{-1}(x))$ is the irreducible polynomial of degree $r$  over $\Bbb F_{2^n}$  and has a root $\beta$ in an extension with $$\beta=(\widetilde {M}, \tau)\alpha =M(\tau(\alpha))=\frac{a\tau(\alpha)+b}{c\tau(\alpha)+d}.  $$

\subsection{ Equivalent extended Goppa codes}
Let $g(x)$ be irreducible of degree $r$ over $\Bbb F_{2^n}$ and  $L=(\alpha_0,\ldots,\alpha_{2^n-1})$ an ordered subset of $2^n$ distinct points in  ${\mathbb{F}}_{2^n}$. Suppose that $g(x)$ has a root $\alpha$ in an extension over $\Bbb F_{2^n}$.
Then
 $$\Gamma(g, L)=\{x=(x_0,x_1,\ldots,x_{2^n-1})\in \Bbb F_2^{2^n}: H(\alpha)x^T=0\}, $$ where the parity check matrix is $$H(\alpha)=(\frac{1}{\alpha-\alpha_0}~\frac{1}{\alpha-\alpha_1}\ldots \frac{1}{\alpha-\alpha_{2^n-1}}).$$ See more detail in \cite{3}.

For convenience, we will denoted by
$$C(\alpha)=\Gamma (g, L)\ \ \and \ \overline {C(\alpha)}=\overline {\Gamma}(g, L),$$
where $\overline {\Gamma}(g, L)$ is the extended Goppa code of $\Gamma(g, L)$.

The following Lemma is similar to that in \cite{1}. For completeness, we give the proof.

\begin{lem}\label{equ}
Let $L=(\alpha_0,\ldots,\alpha_{2^n-1})$ and $\overline L=(\alpha_0,\ldots,\alpha_{2^n-1},\infty)$ be  ordered tuples of $2^n$ and $2^n+1$ distinct points in the projective line $\overline{\mathbb{F}}_{2^n}$, respectively.
Let $\alpha$ be  a root of an irreducible polynomial $g(x)=\sum_{i=0}^r g_ix^i$  of degree $r$ over $\Bbb F_{2^n}$. Let $M=\left(\begin{array}{cc} a&b\\c&d\end{array}\right)\in GL_2(\Bbb F_{2^n})$, $\tau\in G=Gal(\Bbb F_{2^{6n}}/\Bbb F_2)$, $(\widetilde M,\tau)\in P\Gamma L_2(\Bbb F_{2^n})$,
 $\overline L''=(M(\tau(\alpha_0)),\ldots,M(\tau(\alpha_{2^n-1})),M(\infty))$,
$\tau g(x)=\sum_{i=0}^r\tau(g_i)x^i$,  and $(\tau g)'(x)=(-cx+a)^r\cdot \tau g(M^{-1} (x))$ a root $\beta=(\widetilde M, \tau)\alpha=\frac{a\tau(\alpha)+b}{c\tau(\alpha)+d}$. Then  the Alternant code  $\mathcal A_{r+1}(v_{g, \overline L}, \overline L)$ is equal to the Alternant code $\mathcal A_{r+1}(v_{(\tau g)', \overline L''}, \overline L'')$ and
the extended Goppa code $\overline{C(\alpha)}$ is permutation equivalent to the extended Goppa code $\overline{C(\beta)}$, denoted by  $\overline{C(\alpha)}\cong \overline{C(\beta)}$.
 \end{lem}
\begin{proof} Let $\overline L=(\alpha_0,\ldots,\alpha_{2^n-1},\infty)=L\cup \{\infty\}$ be an ordered tuple of $2^n+1$ distinct points in the projective line $\overline{\mathbb{F}}_{2^n}=\Bbb F_{2^n}\cup \{\infty\}$ and  $g(x)$  irreducible of degree $r$ over $\Bbb F_{2^n}$.  Then by Remark \ref{extn},
 $$\mathcal A_{r+1}(v_{g, \overline L}, \overline L)=\overline {\Gamma}( g,L),$$
 where $\mathcal A_{r+1}(v_{g, \overline L}, \overline L)$ is the Alternant code and $\overline {\Gamma}(g,L)$ is the extended code of $\Gamma(g,L)$.
 Let $M=\left(\begin{array}{cc} a&b\\c&d\end{array}\right)\in GL_2(\Bbb F_{2^n})$, $\tau\in G=Gal(\Bbb F_{2^{6n}}/\Bbb F_2)$, $(\widetilde M,\tau)\in P\Gamma L_2(\Bbb F_{2^n})$,  $\tau g(x)=\sum_{k=0}^r\tau(g_k)x^k$,  $(\tau g)'(x)=(-cx+a)^r\cdot \tau g(M^{-1} (x))$ a root $\beta=(\widetilde M, \tau)\alpha=\frac{a\tau(\alpha)+b}{c\tau(\alpha)+d}$, and $$\overline L''=(\widetilde M, \tau)(\overline L)=(M(\tau(\alpha_0), M(\tau(\alpha_1)), \ldots, M(\tau(\alpha_{2^n-1}), M(\infty)).$$
 Then by Corollary \ref{2.777},
 \begin{equation}\label{2.1}\mathcal A_{r+1}(v_{g, \overline L}, \overline L)=\mathcal A_{r+1}(v_{(\tau g)', \overline L''}, \overline L''). \end{equation}
 Let $g(x)$ have a root $\alpha$ in an extension over $\Bbb F_{2^n}$. Then
$(\tau g)'(x)$
 is irreducible over $\Bbb F_{2^n}$ and  has a root  $\beta=(\widetilde M, \tau)\alpha=\frac{a\tau(\alpha)+b}{c\tau(\alpha)+d}$ in an extension over $\Bbb F_{2^n}$. Note that $\overline L$  and $\overline L''$ arrange in  distinct orders of the projective set $\overline {\Bbb F}_{2^n}$.  Hence by Remark \ref{extn} and (\ref{2.1}),
 $$\overline {\Gamma}(g, L)\cong \overline {\Gamma}((\tau g)', L),\ i.e.,\ \overline {C(\alpha)}\cong \overline {C(\beta)}, $$
i.e., the extended Goppa code $\overline{C(\alpha)}$ is permutation equivalent to the extended Goppa code $\overline{C(\beta)}$.
\end{proof}

In \cite{1}, Ryan used Lemma \ref{equ} to give an upper bound on the number of extended irreducible binary quartic Goppa codes of length $2^n+1$, where $n(>3)$ is a prime number.
In this paper, we shall obtain an upper bound on the number of extended irreducible binary sixtic quatic Goppa codes of length $2^n+1$, where $n(>3)$ is a prime number.
By $3|(2^n+1)$ for any $n(>3)$ a prime, we prove that there are orbits of length $3$ in Propositions
\ref{9231},  which is important and different from \cite{1}.



\section{  prime $n>3$}
In this section, we investigate irreducible binary sextic Goppa codes over $\mathbb{F}_{2^n}$. We always assume that $n (>3)$ is a prime number.

\begin{defn}
The set $\mathbb{S}=\mathbb{S}(n,6)$ is the set of all elements in $\mathbb{F}_{2^{6n}}$ of degree $6$ over $\mathbb{F}_{2^n}$.
\end{defn}

 In fact, $\mathbb{S}=\mathbb{F}_{2^{6n}}\setminus \{\mathbb{F}_{2^{2n}}\cup \mathbb{F}_{2^{3n}}\}$, then $$| \mathbb{S}|=2^{6n}-2^{2n}-2^{3n}+2^n.$$

\begin{lem}\label{mainl} The projective semi-linear group  $P\Gamma L_2(\Bbb F_{2^n})$ acts on
 on the set $\mathbb{S}$:
 \begin{eqnarray*}\pi: P\Gamma L_2(\mathbb{F}_{2^n}) \times \mathbb{S} &\rightarrow& \mathbb{S}\\
  ((\widetilde{M},\sigma^i), \alpha) &\mapsto& (\widetilde{M},\sigma^i)(\alpha)=\frac{a\alpha^{2^i}+b}{c\alpha^{2^i}+d}=\beta,  \end{eqnarray*}
 where  $M=\left(\begin{array}{cc} a&b\\c&d\end{array}\right)\in GL_2(\Bbb F_{2^n})$, $G=Gal(\Bbb F_{2^{6n}}/\Bbb F_2)=\langle \sigma\rangle$.

 Moreover, by Lemma \ref{equ} the extended Goppa code $\overline{C(\alpha)}$ is permutation  equivalent to the extended Goppa code $\overline{C(\beta)}$, i.e., $\overline {C(\alpha)}\cong \overline {C(\beta)}$.
\end{lem}

By Lemma \ref{mainl}, there is an equivalent relation $\sim$ in $\Bbb S$:  for $\alpha, \beta\in \Bbb  S$,
$$\alpha \sim \beta \Longleftrightarrow \exists (\widetilde M, \sigma^i)\in P\Gamma L_2(\mathbb{F}_{2^n}),\ (\widetilde M, \sigma^i)\alpha =\beta;$$
moreover, by Lemma \ref{equ}
\begin{equation}\label{3.1} \alpha\sim \beta
\Longrightarrow \overline {C(\alpha)}\cong \overline {C(\beta)}. \end{equation}

There is an interesting question:
does the converse proposition of (\ref{3.1}) hold?

By Lemma \ref{mainl}, the orbit of $\alpha\in \Bbb S$ is
\begin{eqnarray}\label{3.2}&&\Omega_{\alpha}=\{\beta\in \Bbb S:\alpha \sim \beta\}\\ &=&\left\{(\widetilde M, \sigma^i)\alpha=\frac{a\alpha^{2^i}+b}{c\alpha^{2^i}+d}: \forall M=\left(\begin{array}{cc} a&b\\c&d\end{array}\right)\in GL_2(\Bbb F_{2^n}), 0\le i\le 6n-1\right\} \nonumber\\  &=&\left\{\sigma^i(\widetilde M( \alpha))=(\frac{a\alpha+b}{c\alpha +d})^{2^i}: \forall M=\left(\begin{array}{cc} a&b\\c&d\end{array}\right)\in GL_2(\Bbb F_{2^n}), 0\le i\le 6n-1\right\}.\nonumber\end{eqnarray}
Then
 the number of different orbits in $\Bbb S$ under the action of the group $P\Gamma L_2(\Bbb F_{2^n})$ is the upper bound of inequivalent extended irreducible binary sixtic   Goppa codes.

To compute the number of orbits in $\Bbb S$ under the action of   the group $P\Gamma L_2(\Bbb F_q)$, by (\ref{3.2}) we shall divide the action into two actions.

(1) We consider the projective linear group $PGL_2(\Bbb F_{2^n})$ acting on $\Bbb S$. Then
\begin{equation}\label{3.32}\Bbb S=\bigcup_{\alpha \in I}^{\cdot} O_{\alpha},\ \Omega=\{O_{\alpha}: \alpha\in I\}, \end{equation}
where $\Omega$ is the set of all distinct orbits in $\Bbb S$ under the action of the projective linear group $PGL_2(\Bbb F_{2^n})$. In fact, $\Omega$ is a partition of $\Bbb S$.

(2) We consider the Galois group $G=Gal(\Bbb F_{2^{6n}}/\Bbb F_2)=\langle \sigma \rangle $ acting on $\Omega$. Then there is a partition of $\Omega$:
\begin{equation}\label{3.33}\Omega=\bigcup_{\alpha\in J}^{\cdot}\Omega_{\alpha},\end{equation}  $$\Omega_{\alpha}=G(O_{\alpha})=\{\sigma^i(O_{\alpha}): 0\le i\le 6n-1\}\subset \Omega,$$
where $\{\Omega_{\alpha}: \alpha\in J\}$ is the set of all distinct orbits in $\Bbb S$ under the action of the  projective semi-linear group $P\Gamma L_2(\Bbb F_{2^n})$. In fact, $|J|$ is just the number of different orbits in $\Bbb S$ under the action of the group $P\Gamma L_2(\Bbb F_{2^n})$.

In order to
   count the size  of  the set $J$, we shall use
 Cauchy Frobenius Theorem about orbits in \cite{martin}.
 \begin{lem}\label{C}
Let $G$ be a finite group acting on a set $X$. For any $g \in G$, let $F(g)$ denote the set of elements of $X$ fixed by $g$. Then the number of distinct orbits in $X$ under the action of the group $G$ is $\frac{1}{| G |}\sum\limits_{g \in G}| F(g)| $.
\end{lem}

We firstly state the main result in the following theorem.

\begin{thm}\label{725}
Let $n>3$ be a prime number. The number of extended irreducible binary sextic Goppa codes of length $2^n+1$ over $\Bbb F_{2^n}$ is at most $\frac{2^{3n}+2^{2n}+3\cdot 2^n+12n-18}{6n}$.
\end{thm}

To prove Theorem \ref{725}, we shall show some propositions.

\subsection{ $PGL_2(\mathbb{F}_{2^n}) \times \mathbb{S} \rightarrow \mathbb{S}$ }\label{sub3.2}

Consider the projective linear group $PGL_2(\mathbb{F}_{2^n})$ acting on the set $\mathbb{S}$ as follows:
 \begin{eqnarray*} PGL(\mathbb{F}_{2^n}) \times \mathbb{S} &\rightarrow& \mathbb{S}\\
  (\widetilde{M}, \alpha) &\mapsto& \widetilde{M} (\alpha)= M(\alpha)=\frac{a\alpha+b}{c\alpha+d},\end{eqnarray*}
 where $M=\left(
                               \begin{array}{cc}
                                 a & b \\
                                 c & d \\
                               \end{array}
                             \right)\in GL(\mathbb{F}_{2^n})$. Clearly, it is a faithful action. Then
                             $$\Bbb S=\bigcup_{\alpha \in I}^{\cdot}O_{\alpha}, \Omega=\{O_{\alpha}: \alpha\in I\},$$
where $\Omega$ is the set of all distinct orbits in $\Bbb S$ under the action of the projective linear group $PGL_2(\Bbb F_{2^n})$.  In fact,   $\Omega $ is a partition of $\Bbb S$. We shall calculate the size of the set  $I$.

It is clear that $$|PGL_2(\Bbb F_{2^n})|=\frac{|GL_2(\Bbb F_{2^n})|}{|\{aE_2: a\in \Bbb F_{2^n}^*\}|}=2^{3n}-2^n.$$
 For $\alpha\in I$,
let $H_{\alpha}$ denote the stabilizer of $\alpha$ under the action of  $PGL_2(\mathbb{F}_{2^n})$, then $$H_{\alpha}=\left\{\widetilde M(\alpha)=\alpha: \widetilde M\in PG L_2(\Bbb F_{2^n})\right\}=\left\{\widetilde {E_2}\right\}$$ and $$|O_{\alpha}|=\frac{|PGL_2(\Bbb F_{2^n})|}{|H_{\alpha}|}=2^{3n}-2^n.$$ By $|\Bbb S|=2^{6n}-2^{2n}-2^{3n}+2^n$,
$$|I|=\frac{|\Bbb S|}{|O_{\alpha}|}=2^{3n}+2^{n}-1. $$

 Moreover, we investigate the structure of each orbit $O_{\alpha}$. Note that the projective linear group $PGL_2(\Bbb F_{2^n})$ acts transitively on the orbit $O_{\alpha}$.

Since there is an injective homomorphism of two groups:
\begin{eqnarray*}AGL_2(\Bbb F_{2^n})&\rightarrow& PGL_2(\Bbb F_{2^n})\\
M=\left(\begin{array}{cc}e &f \\ 0 & 1\end{array}\right)&\mapsto & \widetilde M, \end{eqnarray*}
the affine group $AGL_2(\Bbb F_{2^n})$ is viewed as the subgroup of $PGL_2(\Bbb F_{2^n})$.

 Hence there is an action of $AGL_2(\Bbb F_{2^n})$ on $O_{\alpha}$:
 \begin{eqnarray*} AGL_2(\mathbb{F}_{2^n}) \times O_{\alpha} &\rightarrow& O_{\alpha}\\
  (M, \beta) &\mapsto&  e\beta+f,\end{eqnarray*}
 where $M=\left(\begin{array}{cc}e &f \\ 0 & 1\end{array}\right)\in AGL_2(\Bbb F_{2^n})$,  $\beta=\frac{a\alpha+b}{c\alpha+d}\in O_{\alpha}, a,b,c,d\in \mathbb{F}_{2^n},ad-bc\neq 0$. We shall investigate all distinct orbits in $O_{\alpha}$ under the action of the affine group $AGL_2(\Bbb F_{2^n})$.

For $\beta=\frac{a\alpha +b}{c\alpha +d}\in O_{\alpha}$, the orbit of $\beta$ in $O_{\alpha}$ under the action of the affine group $AGL_2(\Bbb F_{2^n})$ is
$$A({\beta})=\left\{M(\beta)=e\beta+f: M=\left(\begin{array}{cc}e &f \\ 0 & 1\end{array}\right)\in AGL_2(\Bbb F_{2^n})\right\}.$$

The following lemma is from \cite{26,1}.

\begin{lem}
$$O_{\alpha}=(\bigcup\limits_{\gamma\in \mathbb{F}_{2^n}}^{\cdot} A(\frac{1}{\alpha+\gamma}))\bigcup\limits ^{\cdot}A(\alpha),$$
which are disjoint unions. Moreover, there is a partition of $O_{\alpha}:$
\begin{equation} \label{3.4}\Delta_{\alpha}=\left\{A(\alpha), A(\frac 1{\alpha+\gamma}): \gamma\in \Bbb F_{2^n}\right\}.\end{equation}
\end{lem}

\subsection{$G \times \Omega \rightarrow \Omega$ }
Let $G=Gal(\Bbb F_{2^{6n}}/\Bbb F_2)=\langle \sigma\rangle$ be the Galois group of $\Bbb F_{2^{6n}}$ over $\Bbb F_2$ and $\Omega$ defined as (3.3).
Now consider the group $G$ acting on the set $\Omega$:
 \begin{eqnarray*}\varphi: G \times \Omega &\rightarrow& \Omega\\
  (\sigma^i, O_{\alpha}) &\mapsto& \sigma^i( O_{\alpha})=O_{\sigma^i (\alpha)}.\end{eqnarray*}
Now we shall  count the number of distinct orbits in $\Omega$ under the action of $G$.   By Lemma \ref{C}, we need to calculate that for $\sigma^i\in G$, $0\le i\le 6n-1$,
\begin{equation}\label{3.6} |F(\sigma^i)|=\left|\left\{O_{\alpha}\in \Omega: \sigma^i(O_{\alpha})=O_{\alpha}\right\}\right|. \end{equation}

\begin{rem}\label{725r}
(1) In fact, for $0 \le i\le 6n-1$,   let $d=\gcd(6n, i)$, then $\langle \sigma^i\rangle=\langle \sigma^d\rangle$. Thus,
$$\{O_{\alpha}\in\Omega:\sigma^i(O_{\alpha})=O_{\alpha}\}
=\{O_{\alpha}\in\Omega:\sigma^d(O_{\alpha})=O_{\alpha}\}$$
and $$|F(\sigma^i)|=|F(\sigma^d)|.$$

(2) Let $o(\sigma^i)$ denote the order of $\sigma^i$ in $G$, then $o(\sigma^i)\mid|G|$. By $|G|=6n$ and $n$ a prime,  $$o(\sigma^i)\in\{1, 2, 3, 6, n, 2n, 3n, 6n\}.$$
Suppose that $o(\sigma^i)\in\{6n, 3n\}$ and $\sigma^i(O_{\alpha})=O_{\alpha}$. Then $H=\langle \sigma^i\rangle=\langle \sigma^k\rangle$, $k\in \{1,2\}$,   and $\sigma^k(O_{\alpha})=O_{\alpha}$.
 Hence $\frac{a\alpha+b}{c\alpha +d}=\alpha^{2^k}$, which is impossible by $\alpha$ of degree $6$ over $\Bbb F_{2^n}$.
Therefore $$|F(\sigma^i)|=0\mbox{ if }o(\sigma^i)\in\{6n, 3n\}.$$

In the next section, we shall discuss other cases.
\end{rem}

\section{The action of $G$ on $\Omega$}
\begin{prop}\label{7254.1} Let $ G=Gal(\Bbb F_{2^{6n}}/\Bbb F_{2})=\langle \sigma\rangle $.

(1) If $o(\sigma^i)=1$, then  $|F(\sigma^i)|=|\Omega| =2^{3n}+2^n-1$.

(2) If $o(\sigma^i)=2$, then $|F(\sigma^{i})|=2^{2n}-1$.
\end{prop}
\begin{proof}
(1)  If $o(\sigma^i)=1$, then for all $O_{\alpha}\in \Omega$,
 $\sigma^{i}(O_{\alpha})=O_{\alpha}$ and  $|F(\sigma^i)|=|\Omega| =2^{3n}+2^n-1$.

(2)
By $o(\sigma^i)=2$, $H=\langle \sigma^i\rangle =\langle \sigma^{3n}\rangle $.
Suppose that  $\sigma^{3n}(O_{\alpha})=O_{\alpha}$ for $O_{\alpha}\in \Omega$. Then   the group $H$ acts on the partition $\Delta_{\alpha}$ of $O_{\alpha}$ in  (\ref{3.4}) as follows:
  \begin{eqnarray*} H\times \Delta_{\alpha}&\longrightarrow& \Delta_{\alpha}\\
(\sigma^{3ni},A(\alpha))&\longmapsto& \sigma^{3ni}(A(\alpha)),i=0,1.\end{eqnarray*} Denote by $O_{A(\alpha)}$ the orbit of $A(\alpha)$ under the action of the subgroup  $H=\langle\sigma^{3n}\rangle$. Then there is a class equation:
 $$2^n+1=|\Delta_{\alpha}|=\sum |O_{A(\alpha)}|.$$
 By $|H|=2$,  there is $A(\alpha)\in \Delta_{\alpha}$ such that $|O_{A(\alpha)}|=1$, i.e.,
   $\sigma^{3n}(A(\alpha))=A(\alpha)$.
Hence
$$\sigma^{3n}(\alpha)=\alpha^{2^{3n}}=c_1\alpha+c_2,$$
 with $c_1\in \mathbb{F}_{2^n}^*, c_2 \in \mathbb{F}_{2^n}.$ Then $\alpha=\sigma^{6n}(\alpha)=c_1\sigma^{3n}(\alpha)+c_2=c_1^2\alpha+c_1c_2+c_2$. Thus $c_1=1$
and $$\alpha^{2^{3n}}=\alpha+c_2$$ with $c_2\neq 0$ (if $c_2= 0$, then $\alpha\in \Bbb F_{2^{3n}}$ which is contradictory.). Consequently, $$(c_2^{-1}\alpha)^{2^{3n}}+c_2^{-1}\alpha+1=0.$$
If $c_2^{-1}\alpha$ is viewed as $\alpha$
 then  $\alpha$ satisfies the equation: $$x^{2^{3n}}+x+1=0.$$
Hence there is a factorization:
 $$x^{2^{3n}}+x+1=\prod\limits_{c\in \mathbb{F}_{2^{3n}}}(x+\alpha+c). $$
It is clear  that $\alpha^{2^{6n}}=\alpha$  and $\alpha^{2^{3n}}=\alpha+1$, so all roots  $\alpha+c\in \Bbb F_{2^{6n}} \setminus\Bbb F_{2^{3n}}$, $c\in\Bbb F_{2^{3n}}$.

Moreover, if $\alpha$ is a root of the polynomial $x^{2^{3n}}+x+1$ and $\alpha\in \Bbb F_{2^{2n}}$,  then $\alpha$ is a root of the polynomial $x^{2^n}+x+1$ and there is a factorization:
$$x^{2^n}+x+1=\prod_{c\in\Bbb F_{2^n}}(x+\alpha+c)$$
and all roots $\alpha+c\in\Bbb F_{2^{2n}}$, $c\in \Bbb F_{2^n}$.

Hence there are  $2^{3n}-2^n$ roots of $x^{2^{3n}}+x+1$ in $\Bbb S$.

Conversely, if $\alpha\in \Bbb S$ is a root of $x^{2^{3n}}+x+1$, then $\sigma^{3n}(A(\alpha))=A(\alpha)$.

Therefore $\alpha\in \Bbb S$ is a root of $x^{2^{3n}}+x+1=0$ if and only if $\sigma^{3n}(A(\alpha))=A(\alpha)$; moreover, $A(\alpha)$ has $2^n$ roots $\alpha+c, c\in \mathbb{F}_{2^n}$, of the polynomial $x^{2^{3n}}+x+1$.

  Suppose that $\sigma^{3n} (A(\alpha))=A(\alpha)$, $\alpha^{2^{3n}}+\alpha+1=0$, and $\sigma^{3n}(A(\frac{1}{\alpha+\gamma}))=A(\frac{1}{\alpha+\gamma})$. Then
    $$(\frac 1 {\alpha+\gamma})^{2^{3n}}=\frac 1{\alpha+1+\gamma
  }=\frac {b_1}{\alpha +\gamma}+b_2, 0\ne b_1, b_2\in \Bbb F_{2^n},$$
 which is  contradictory with  $\alpha$  of degree $6$ over $\mathbb{F}_{2^n}$.

 In conclusion, $\sigma^{3n}(O_{\alpha})=O_{\alpha}$ if and only if $\sigma^{3n}(A(\alpha))=A(\alpha)$ if and only if there are $2^n$ roots of $x^{2^{3n}}+x+1$ in $O_{\alpha}$.
Then $$|F(\sigma^{3n})|=\frac{2^{3n}-2^n}{2^n}=2^{2n}-1.$$
\end{proof}

\begin{lem} \label{4.20} $$\left| \left\{\beta\in \mathbb{S}: \beta^{2^{2n}}+\frac{1}{\beta}+1=0 \right\}\right|=2^{2n}-2^n-2.$$
\end{lem}
\begin{proof} Now we investigate  roots  of the following equation
 \begin{equation}\label{4.11} x^{2^{2n}}+\frac1x +1=0.\end{equation}

Suppose that $\beta$ is a root in (4.1) and \begin{equation}B=\left(\begin{array}{ll}1&1\\1&0\end{array}\right).\end{equation} Then $\beta^{2^{2n}}=B(\beta)$. By $o(B)=3$,  $\beta^{2^{6n}}=B^3(\beta)=\beta$ and $\beta\in \Bbb F_{2^{6n}}$. In the following, we shall find all roots $\beta$ in (\ref{4.11}) such that $\beta\in \Bbb F_{2^{2n}}\cup \Bbb F_{2^{3n}}$.

It is clear that
   $\beta\in \mathbb{F}_{2^{2n}}$ if and only if $\beta=\frac{1}{\beta}+1$, i.e., $o(\beta)=3$ and $\beta\in \Bbb F_{4}$. Then
$$\left| \left\{\beta\in \mathbb{F}_{2^{2n}}:\beta^{2^{2n}}+\frac{1}{\beta}+1=0 \right\}\right|=\left|\left\{\beta\in \mathbb{F}_{2^{2n}}:\beta+\frac{1}{\beta}+1=0 \right\}\right|=2.$$

If
   $\beta\in \mathbb{F}_{2^{3n}}$, then  $\beta^{2^n}+\frac{1}{\beta+1}=0$.
 Conversely, if $\beta^{2^n}=\left(\begin{array}{ll}0 &1\\1&1\end{array}\right) (\beta)$, then $\beta^{2^{3n}}=\left(\begin{array}{ll}0 &1\\1&1\end{array}\right)^3(\beta)=\beta$ and $\beta\in\Bbb F_{2^{3n}}$.
    Hence
$$\left|\left\{\beta\in \mathbb{F}_{2^{3n}}:\beta^{2^{2n}}+\frac{1}{\beta}+1=0 \right\}\right|=\left|\left\{\beta\in \mathbb{F}_{2^{3n}}:\beta^{2^{n}}({\beta}+1)+1=0 \right\}\right|=2^n+1.$$

Therefore $$\left|\left\{\beta\in \mathbb{S}: \beta^{2^{2n}}+\frac{1}{\beta}+1=0 \right\}\right|=2^{2n}+1-2-(2^n+1)=2^{2n}-2^n-2.$$
\end{proof}
\begin{prop}\label{9231} Let $ G=Gal(\Bbb F_{2^{6n}}/\Bbb F_{2})=\langle \sigma\rangle $. If $o(\sigma^i)=3$, then  $|F(\sigma^{i})|=2^n-2$.
\end{prop}
\begin{proof}
By $o(\sigma^i)=3$, $H=\langle \sigma^i\rangle =\langle \sigma^{2n}\rangle $.
Suppose that  $\sigma^{2n}(O_{\alpha})=O_{\alpha}$ for $O_{\alpha}\in \Omega$. Then   the group $H$ acts on the partition $\Delta_{\alpha}$ of $O_{\alpha}$ in  (\ref{3.4}) as follows:
  \begin{eqnarray*} H\times \Delta_{\alpha}&\longrightarrow& \Delta_{\alpha}\\
(\sigma^{2ni},A(\alpha))&\longmapsto& \sigma^{2ni}(A(\alpha)),i=0,1,2.\end{eqnarray*} Denote by $O_{A(\alpha)}$ the orbit of $A(\alpha)$ under the action the subgroup  $H=\langle\sigma^{2n}\rangle$.
 For $A(\alpha)\in \Delta_{\alpha}$,   $|O_{A(\alpha)}|=1$ or $3$ by $|H|=3$.

Suppose that there is $A(\alpha)\in \Delta_{\alpha}$ such that  $| O_{A(\alpha)}| =1$. Then $\sigma^{2n}(A(\alpha))=A(\alpha)$. By the proof of Proposition \ref{7254.1},
 $\alpha$ is a root of $x^{2^{2n}}+x+1=0$ and  $\alpha^{2^{4n}}=\alpha$,  which is a contradiction with $\alpha$  of degree $6$ over $\mathbb{F}_{2^n}$.  Hence  $| O_{A(\alpha)}| = 3$ for all $A(\alpha)\in \Delta_{\alpha}$.

 Without loss of generality, let $\sigma^{2n}(A(\alpha))=A(\frac 1{\alpha})$.
 Then
 $$\alpha^{2^{2n}}=\frac a {\alpha}+b=B_0(\alpha),\alpha^{2^{6n}}=B_0^3\alpha=\alpha,$$
where $B_0=\left(\begin{array}{cc} b & a\\ 1& 0\end{array}\right)\in  GL_2(\Bbb F_{2^n})$. Hence $B_0^3=\left(\begin{array}{cc} b^3 & ab^2+a^2\\ b^2+a& ab\end{array}\right)$ is a scalar matrix and $ a=b^2\in \Bbb F_{2^n}^*$. If $\frac {\alpha}{b}$ is viewed as $\alpha$,
then  $\alpha\in \Bbb S$ satisfies the equation (4.1). Hence
$\alpha$, $\alpha^{2^{2n}}=B(\alpha)=\frac 1{\alpha}+1$, $\alpha^{2^{4n}}=B^2(\alpha)=\frac 1{\alpha+1}$ are roots of (4.1), where $B$ is defined as (4.2). So
$\sigma^{2n}(A(\alpha))=A(\frac{1}{\alpha}+1)=A(\frac 1{\alpha})$, $\sigma^{2n}(A(\frac{1}{\alpha}))=A(\frac{1}{\alpha+1})$, and $\sigma^{2n}(A (\frac{1}{\alpha+1})) = A(\alpha)$.

Suppose that $\alpha\in A(\alpha)$ is a solution of (4.1) and   $P(\alpha)=c\alpha +d\in A(\alpha)$ is also a solution of  (\ref{4.11}), where $P=\left(\begin{array}{ll} c &d\\ 0&1\end{array}\right)$.  Then by $\alpha^{2^{2n}}=B(\alpha)$, $BP(\alpha)=(P(\alpha))^{2^{2n}}=P(\alpha)^{2^{2n}}=PB(\alpha)$ and
$$B^{-1}P^{-1}BP=\frac 1 a\left(\begin{array}{cc}c^2&cd\\c+c^2+cd&d+1+d^2+cd\end{array}\right)$$
is a scalar matrix, i.e., $c=1, d=0$. Hence there is a unique solution $\alpha\in A(\alpha)$ of (\ref{4.11}), so there are $2^n+1$ roots of $x^{2^{2n}}+x+1$ in $O_{\alpha}$.

Conversely, if $\alpha$ is a solution of (\ref{4.11}), then $\sigma^{2n}(O_{\alpha})=O_{\alpha}$.

Consequently, by Lemma \ref{4.20} $$|F(\sigma^{2n})|=\frac{2^{2n}-2^n-2}{2^n+1}=2^n-2.$$
\end{proof}

\begin{prop}\label{923}
 Let $ G=Gal(\Bbb F_{2^{6n}}/\Bbb F_{2})=\langle \sigma\rangle $. If $o(\sigma^i)=6$, then  $|F(\sigma^{i})|=0$.
\end{prop}
\begin{proof} By $o(\sigma^i)=6$, $H=\langle \sigma^i\rangle =\langle \sigma^{n}\rangle $.
Suppose that  $\sigma^{n}(O_{\alpha})=O_{\alpha}$ for $O_{\alpha}\in \Omega$. Then   the group $H$ acts on the partition $\Delta_{\alpha}$ of $O_{\alpha}$ in  (\ref{3.4}) as follows:
  \begin{eqnarray*} H\times \Delta_{\alpha}&\longrightarrow& \Delta_{\alpha}\\
(\sigma^{ni},A(\alpha))&\longmapsto& \sigma^{ni}(A(\alpha)),i=0,1,\ldots,5.\end{eqnarray*} Denote by $O_{A(\alpha)}$ the orbit of $A(\alpha)$ under the action the subgroup  $H=\langle\sigma^{n}\rangle$. Then there is a class equation:
 \begin{equation} \label{4.31}2^n+1=|\Delta_{\alpha}|=\sum |O_{A(\alpha)}|.\end{equation}
For $A(\alpha)\in \Delta_{\alpha}$,  $|O_{A(\alpha)}|\in \{1, 2, 3, 6\}$ by $|H|=6$.

Suppose that there is $A(\alpha)\in O_{\alpha}$ such that  $|O_{A(\alpha)}|=1$. Then $\sigma^n(A(\alpha))=A(\alpha)$ and $\alpha^{2^n}=a\alpha+b$.  Hence $\alpha=\alpha^{2^{6n}}=a^6\alpha$, $a=1$, and $\alpha^{2^{2n}}=\alpha$, which is contradictory.

Suppose that there is $A(\alpha)\in O_{\alpha}$ such that  $|O_{A(\alpha)}|=3$. Without loss of generality, let $\sigma^{n}(A(\alpha))=A(\frac 1{\alpha})$.
Then
 $$\alpha^{2^{n}}=\frac a {\alpha}+b=B_0(\alpha),\alpha^{2^{6n}}=B_0^6\alpha=\alpha,$$
where $B_0=\left(\begin{array}{cc} b & a\\ 1& 0\end{array}\right)\in  GL_2(\Bbb F_{2^n})$. Hence $$B_0^6=\left(\begin{array}{cc} b^6+(ab^2+a^2)(b^2+a) & (ab^2+a^2)(b^3+ab)\\ (b^2+a)(b^3+ab)& (b^2+a)(ab^2+a^2)+a^2b^2\end{array}\right)$$ is a scalar matrix and $ a=b^2\in \Bbb F_{2^n}^*$ or $b=0$.

If $b=0$, then $B_0^2=\left(\begin{array}{cc} a & 0\\ 0& a \end {array}\right)$ and $\alpha\in \Bbb F_{2^{2n}}$, which is a contradiction.

If $a=b^2\in \Bbb F_{2^n}^*$, then $(\frac{\alpha}{b})^{2^{n}}=\frac{b}{\alpha}+1=
B(\frac{\alpha}{b})$ with $B=\left(\begin{array}{cc}1 & 1 \\1 & 0 \\\end{array} \right)$.
Since $(\frac{\alpha}b)^{2^{3n}}= B^3(\frac{\alpha}b)=\frac{\alpha}b$,  $\frac{\alpha}{b}\in \Bbb F_{2^{3n}}$, which is a contradiction.

  Hence $| O_{A(\alpha)}| $ can not be $1$ and $3$.

  Suppose that $|O_{A(\alpha)}| =2$ or $6$ for all $A(\alpha)\in O_{\alpha}$.  Then it is a contradiction in (\ref{4.31}).

 Therefore, $$|F(\sigma^{n})|=0.$$
\end{proof}
\begin{lem}\label{4.51} If $\alpha\in \Bbb S$ and $\alpha^{2^6}=\alpha+b$, $b\in \Bbb F_{2^n}$. Then there is an element $c\in\Bbb F_{2^n}$ such that $(\alpha+c)^{2^6}=(\alpha+c)$.
\end{lem}
\begin{proof} Let $T^{6n}_{6}: \Bbb F_{2^{6n}}\rightarrow \Bbb F_{2^6}$ and  $T^n_1: \Bbb F_{2^n}\rightarrow \Bbb F_2$ be the trace functions, then $T^{6n}_{6}(b)=0$ by $b=\alpha^{2^6}-\alpha$.
Let $\Bbb Z/(n)$ be the residue class  ring modulo $n$, then by $\gcd(n, 6)=1$   $6\cdot \Bbb Z/(n)=\Bbb Z/(n)$ and $T^n_1(b)=T^{6n}_{6}(b)=0$  for $b\in \Bbb F_{2^n}$. Let  $S=\{c^{2^6}+c: c\in \Bbb F_{2^n}\}$, then $|S|=2^{n-1}$ and $ker(T^n_1)=S$.
Hence there is an element $c\in\Bbb F_{2^n}$ such that $b=c^{2^6}+c$, so $(\alpha+c)^{2^6}=(\alpha+c)$.
\end{proof}

\begin{prop}\label{4.61}
 Let $ G=Gal(\Bbb F_{2^{6n}}/\Bbb F_{2})=\langle \sigma\rangle $. If $o(\sigma^i)=n$, then  $|F(\sigma^{i})|=9$.
\end{prop}
\begin{proof}
By $o(\sigma^i)=n$, $H=\langle \sigma^i\rangle =\langle \sigma^{6}\rangle $.
Suppose that  $\sigma^{6}(O_{\alpha})=O_{\alpha}$ for $O_{\alpha}\in \Omega$. Then   the group $H$ acts on the partition $\Delta_{\alpha}$ of $O_{\alpha}$ in  (\ref{3.4}) as follows:
  \begin{eqnarray*} H\times \Delta_{\alpha}&\longrightarrow& \Delta_{\alpha}\\
(\sigma^{6i},A(\alpha))&\longmapsto& \sigma^{6i}(A(\alpha)),i=0,1,\ldots,n-1.\end{eqnarray*} Denote by $O_{A(\alpha)}$ the orbit of $A(\alpha)$ under the action of  the subgroup  $\langle\sigma^{n}\rangle$. Then there is a class equation:
 $$2^n+1=|\Delta_{\alpha}|=\sum | O_{A(\alpha)}|.$$
 For  $A(\alpha)\in \Delta_{\alpha}$,   $|O_{A(\alpha)}|=1$ or $n$ by $|H|=n$, where $n$ is a prime number.
By   $n\nmid (2^n+1)$,
there exists  $A(\alpha)$ such that $|O_{A(\alpha)}| =1$.

If  $\sigma^{6}(A(\alpha))=A(\alpha)$ for $A(\alpha)\in \Delta_{\alpha}$, then $\sigma^6(\alpha)=a\alpha+b$, $0\ne a, b\in \Bbb F_{2^n}$. Hence $\alpha=\sigma^{6n}(\alpha)=\sigma^{6(n-1)}(a\alpha+b)=a^n\alpha+k$, $k\in \Bbb F_{2^n}$,  so $k=0$, $a^n=1$. Thus $a=1$ by $\gcd(n, 2^n-1)=1$ and
$$\alpha^{2^6}=\alpha+b.$$
By Lemma \ref{4.51}, there exists an element $c\in \Bbb F_{2^n}$ such that $(\alpha+c)^{2^6}=\alpha +c$.  If $\alpha+c$ is viewed as $\alpha$, then there are exactly two roots: $\alpha, \alpha+1$,  of  $x^{2^6}=x$ in $A(\alpha)$.

Moreover, if $\sigma^{6}(A(\frac{1}{\alpha+\gamma}))=A(\frac{1}{\alpha+\gamma})$ and $\alpha^{2^6}=\alpha$. Then by above,
 $$(\frac{a}{\alpha+\gamma}+b)^{2^6}=\frac{a}{\alpha+\gamma}+b, 0\ne a, b\in \Bbb F_{2^n}.$$

Hence $a^{2^6}=a, b^{2^6}=b, \gamma^{2^6}=\gamma$, and $a=1$, $b=0$ or $1$, $\gamma=0$ or $1$ by $n(>3)$ a prime.
So there are exactly six roots: $\alpha,\alpha +1, \frac 1{\alpha}, \frac1{\alpha}+1, \frac1{\alpha+1}, \frac 1{\alpha+1}+1$,  in $O_{\alpha}$ of $x^{2^{6}}=x$.

Conversely, if $\alpha^{2^6}=\alpha$, then $\sigma^6(O_{\alpha})=O_{\alpha}$. Thus, $\left|\Bbb S\bigcap O_{\alpha}\right|=6$.

Therefore, $\sigma^6(O_{\alpha})=O_{\alpha}$ if and only if there are exactly six roots in $O_{\alpha}$ of $x^{2^{6}}=x$.

On the other hand,  $ \mathbb{F}_{2^6}\bigcap \mathbb{S}=\Bbb F_{2^6}\setminus (\Bbb F_{2^2}\cup \Bbb F_{2^3})$ and $\left|\mathbb{F}_{2^6}\bigcap \mathbb{S}\right|=54$.
Then $$|F(\sigma^{6})|=\frac{\left|\mathbb{F}_{2^6}\bigcap \mathbb{S}\right|}{\left|\Bbb S\bigcap O_{\alpha}\right|}=9.$$
\end{proof}

\begin{prop}
 Let $ G=Gal(\Bbb F_{2^{6n}}/\Bbb F_{2})=\langle \sigma\rangle $. If $o(\sigma^i)=2n$, then  $|F(\sigma^{i})|=3$.
\end{prop}
\begin{proof}

By $o(\sigma^i)=2n$, $H=\langle \sigma^i\rangle =\langle \sigma^{3}\rangle $.
Suppose that  $\sigma^{3}(O_{\alpha})=O_{\alpha}$. Then   the group $H$ acts on the partition $\Delta_{\alpha}$ of $O_{\alpha}$ in  (\ref{3.4}) as follows:
  \begin{eqnarray*} H\times \Delta_{\alpha}&\longrightarrow& \Delta_{\alpha}\\
(\sigma^{3i},A(\alpha))&\longmapsto& \sigma^{3i}(A(\alpha)),i=0,1,\ldots,2n-1.\end{eqnarray*} Denote by $O_{A(\alpha)}$ the orbit of $A(\alpha)$ under the action of  $\langle\sigma^{3}\rangle$. Then there is a class equation:
 $$2^n+1=|\Delta_{\alpha}|=\sum | O_{A(\alpha)}|.$$
 For $A(\alpha)\in \Delta_{\alpha}$,  $|O_{A(\alpha)}|\in \{1, 2, n, 2n\}$ by $|H|=2n$.
By  $n\nmid (2^n+1)$,  there exists   $A(\alpha)\in \Delta_{\alpha}$ such that $| O_{A(\alpha)}| =1$ or $2$.

    If there is $A(\alpha)\in \Delta_{\alpha}$ such that $|O_{A(\alpha)}|=2$, then $\sigma^3(A(\alpha))=A(\frac 1{\alpha+\gamma})$ and $\sigma^6(A(\alpha))=A(\alpha)$. Without loss of generality, let $\alpha^{2^6}=\alpha$ and $\alpha^{2^3}=\frac a{\alpha+\gamma}+b, 0\ne a, b\in \Bbb F_{2^n}$. By
    $$\alpha=\alpha^{2^6}=\frac {a^{2^3}}{\frac a{\alpha +\gamma}+b+\gamma^{2^3}}+b^{2^3},$$
    $b=\gamma^{2^3}$,   $a=1$, and $(\alpha+\gamma)^{2^3}=\frac 1{\alpha+\gamma}$. If $\alpha+\gamma$ is viewed as $\alpha$, then $\sigma^3(A(\alpha))=A(\frac 1{\alpha})$ and $\alpha^{2^3}=\frac1{\alpha}$.
   Hence $(\frac1{\alpha+1})^{2^3}=\frac 1{\alpha+1}+1$ and  $|O_{A(\frac 1{\alpha+1})}|=1$.

   Without loss of generality, let
 $A(\alpha)\in\Delta_{\alpha}$ such that  $| O_{A(\alpha)}|=  1$.  Then  $\sigma^3(A(\alpha))=A(\alpha)$ and  $\sigma^6(A(\alpha))=A(\alpha)$, so $\alpha^{2^6}=\alpha$
  by Proposition \ref{4.61} and
    $\alpha^{2^3}=a\alpha +b$, $0\ne a, b\in\Bbb F_{2^n}$. Then $\alpha=\alpha^{2^6}=a^{2^3}(a\alpha+b)+b^{2^3}$, so $a=1$ and  $b=1$ by $\gcd(2^6-1, 2^n-1)=1$. Hence $\alpha^{2^3}+\alpha+1=0$.

    Suppose that there is also $A(\frac 1{\alpha+\gamma})$ such that $|O_{A(\frac 1{\alpha+\gamma})}|=1$. Then
    $$\sigma^3(\frac 1{\alpha+\gamma})=\frac 1{\alpha+1+\sigma^3(\gamma)}=\frac a{\alpha+\gamma}+b, 0\ne a, b\in\Bbb F_{2^n}.$$
  Hence $b=0$, $a=1$, and $\gamma^8+\gamma +1=0$, which is contradictory with $\gamma\in\Bbb F_{2^n}$.

    Hence there is a unique $A(\alpha)\in O_{\alpha}$ such that $|O_{\alpha}|=1$ and there are exactly two roots: $\alpha, \alpha+1$, of $x^{2^3}+x+1=0$ in $O_{\alpha}$.

    Conversely, if $\alpha\in \Bbb S$ is a root of  $x^{2^{3}}+x+1=0$, then $\sigma^{3}(O_{\alpha})=O_{\alpha}$. Thus, $|\Bbb S\bigcap O_{\alpha}|=2$.

    Moreover, it is clear that
    $$\left| \left\{ \alpha\in \Bbb S:  \alpha^{2^{3}}+\alpha+1=0 \right\}\right|=\left| \left\{ \alpha\in \Bbb F_{2^6}\setminus \{\Bbb F_{2^3}\bigcup \Bbb F_{2^2}\}:  \alpha^{2^{3}}+\alpha+1=0 \right\}\right|=6.$$

Therefore, $$|F(\sigma^3)| =\frac{6}{2}=3.$$
\end{proof}

By above propositions, we finally give the proof of Theorem \ref{725} as follows.

{\bf Proof of Theorem \ref{725}} For $i$ an integer, denote by $\phi(i)$ the Euler function. Then, by Remark \ref{725r}(1)
\begin{eqnarray*}
 && \sum\limits_{\sigma^i \in G,i=0,\ldots,6n-1}| F(\sigma^i)|= \sum \limits_{d|6n}|F(\sigma^d)|\phi(\frac{6n}{d})    \\
  &=& |F(\sigma^{0})|\phi(1)+|F(\sigma^{3n})|\phi(2)+|F(\sigma^{2n})|\phi(3)+|F(\sigma^{6})|
\phi(n)+|F(\sigma^{3})|\phi(2n) \\
   &=&  2^{3n}+2^{2n}+3\cdot 2^n+12n-18.
\end{eqnarray*}
By Lemmas \ref{C} and \ref{equ}, there are at most $\frac{2^{3n}+2^{2n}+3\cdot 2^n+12n-18}{6n}$ extended irreducible binary sextic Goppa codes of length $2^n+1$ over $\Bbb F_{2^n}$.

\section{Conclusion}
In this paper, we have obtained an upper bound on the number of extended irreducible Goppa codes of degree $6$ and length $2^n+1$ with $n(>3)$ a prime number. Our results show that many extended Goppa codes become equivalent  and this supports the idea of mounting an enumeration attack on the McEliece cryptosystem using extended Goppa codes.

\end{document}